\documentclass[12pt,letterpaper,showkeys]{article} 
\usepackage{graphicx}
\usepackage{amssymb}
\usepackage{amsmath}
\usepackage{mathrsfs}
\graphicspath{{figures/}}
\newtheorem{theorem}{Theorem}[section] 
\newtheorem{lemma}[theorem]{Lemma} 
 
\newtheorem{proposition}[theorem]{Proposition} 
\newtheorem{remark}[theorem]{Remark}
 
\newtheorem{definition}{Definition}
\newcommand{\beq}{\begin{equation}} 
\newcommand{\eeq}{\end{equation}} 
\newcommand{\beqa}{\begin{eqnarray}} 
\newcommand{\eeqa}{\end{eqnarray}} 
\newcommand{\beqas}{\begin{eqnarray*}} 
\newcommand{\eeqas}{\end{eqnarray*}} 
\newcommand{\ba}{\begin{array}} 
\newcommand{\ea}{\end{array}} 
\newcommand{\bi}{\begin{itemize}} 
\newcommand{\ei}{\end{itemize}} 
\newcommand{\gap}{\hspace*{2em}}

\def\eqnok#1{(\ref{#1})} 
\def\argmin{{\rm argmin}} 
\def\Argmin{{\rm Argmin}}

\def\vgap{\vspace*{.1in}}
\def\QED{\ifhmode\unskip\nobreak\fi\ifmmode\ifinner\else\hskip5pt\fi\fi
  \hbox{\hskip5pt\vrule width5pt height5pt depth1.5pt\hskip1pt}}

\def\card{{\rm Card}} 
\def\diag{{\rm diag}}

\def\eps{{\epsilon}}
\def\cA{{\cal A}} 
 
\def\cO{{\cal O}}
\def\cS{{\cal S}} 
\def\cN{{\cal N}}
\def\cX{{\cal X}}
\def\cXb{{\cX_{\hbeta}}}
\def\cU{{\cal U}} 
\def\cZ{{\cal Z}}
 
\def\bbeta{{\bar \beta}}
\def\hbeta{{\hat \beta}}

\def\hd{{\hat d}}
\def\hD{{\hat D}}

\def\lmax{{\lambda_{\max}}}
 
\def\td{{\tilde d}} 
 
\def\mF{{\mathcal F}} 
\def\mE{{\mathcal E}}
\def\tbeta{{\tilde \beta}}
\def\tx{{\tilde x}}
\def\tu{{\tilde u}}

\def\Arg{{\rm Arg}}
\def\talpha{{\tilde \alpha}}
\def\tbeta{{\tilde \beta}}
\def\tr{{\rm Tr}}
\def\xbuk{{X_{\hbeta_k}(U_k)}}
\def\xhbuk{{X_{\hbeta}(U_k)}}
\title{Smooth Optimization Approach for Sparse \\ Covariance Selection} 
\author{Zhaosong Lu% 
\thanks{
Department of Mathematics, Simon Fraser University, Burnaby, BC, 
V5A 1S6, Canada. (email: {\tt zhaosong@sfu.ca}).
This author was supported in part by SFU President's Research Grant
and NSERC Discovery Grant.}
}

\date{July 3, 2007 (Revised: January, August, September 2008)}
\begin{document}

\maketitle

\begin{abstract} 
In this paper we first study a smooth optimization approach for solving a 
class of non-smooth {\it strictly} concave maximization problems whose 
objective functions admit smooth convex minimization 
reformulations. In particular, we apply Nesterov's smooth optimization 
technique \cite{Nest83-1,Nest05-1} to their dual counterparts that 
are smooth convex problems. It is shown that the resulting approach 
has $\cO(1/{\sqrt{\epsilon}})$ iteration complexity for finding an 
$\epsilon$-optimal solution to both primal and dual problems. We 
then discuss the application of this approach to sparse covariance 
selection that is approximately solved as a $l_1$-norm penalized maximum 
likelihood estimation problem, and also propose a variant of this approach 
which has substantially outperformed the latter one in our computational 
experiments. We finally compare the performance of these approaches with 
other first-order methods, namely, Nesterov's $\cO(1/\epsilon)$ smooth 
approximation scheme and block-coordinate descent method studied in 
\cite{DaOnEl06,FriHasTib07} for sparse covariance selection on a set 
of randomly generated instances. It shows that our smooth optimization 
approach substantially outperforms the first method above, and moreover, 
its variant substantially outperforms both methods above. 

\vskip14pt

 
\vskip14pt


\end{abstract}

\section{Introduction} \label{introduction}

In \cite{Nest83-1,Nest05-1}, Nesterov proposed an efficient smooth 
optimization method for solving convex programming problems of the form
\beq \label{CP}
\min \{ f(u) : u \in U \},
\eeq
where $f$ is a convex function with Lipschitz continuous gradient,  
and $U$ is a closed convex set. It is shown that his method has 
$\cO(1/{\sqrt{\epsilon}})$ iteration complexity bound, where 
$\epsilon>0$ is the absolute precision of the final objective 
function value. A proximal-point-type algorithm for \eqnok{CP} having 
the same complexity as above has also been proposed more recently by 
Auslender and Teboulle \cite{AuTe06-1}. 

Motivated by \cite{DaOnEl06}, we are particularly interested in studying 
the use of smooth optimization approach for solving a class of non-smooth 
{\it strictly} concave maximization problems whose objective functions admit 
smooth convex minimization reformulations in this paper. Our key idea is to 
apply Nesterov's smooth optimization technique \cite{Nest83-1, Nest05-1} to 
their dual counterparts that are smooth convex problems. It is shown that 
the resulting approach has $\cO(1/{\sqrt{\epsilon}})$ iteration complexity 
for finding an $\epsilon$-optimal solution to both primal and dual problems. 

One interesting application of the above approach is for sparse covariance 
selection. Given a set of random variables with Gaussian distribution for 
which the true covariance matrix is unknown, covariance selection is a procedure 
used to estimate true covariance from a sample covariance matrix by 
maximizing its likelihood while imposing a certain sparsity on the 
inverse of the covariance estimation (e.g., see \cite{Demp72}). Therefore, 
it can be applied to determine a robust estimate of the true variance matrix, 
and simultaneously discover the sparse structure in the underlying model. 
Despite its popularity in numerous real-world applications (e.g., see 
\cite{BaElDaNa06,DaOnEl06,YuLi07-1} and the references therein), sparse 
covariance selection itself is a challenging NP-hard combinatorial optimization 
problem. By an argument that is often used in regression techniques such as LASSO 
\cite{Tibs96}, Yuan and Lin \cite{YuLi07-1} and d'Aspremont et al.\ 
\cite{DaOnEl06} (see also \cite{BaElDaNa06}) showed that it can be 
approximately solved as a $l_1$-norm penalized maximum likelihood 
estimation problem. Moreover, the authors of \cite{DaOnEl06} studied 
two efficient first-order methods for solving this problem, that is, 
Nesterov's smooth approximation scheme and block-coordinate descent (BCD)
method. It was shown in \cite{DaOnEl06} that their first method has 
$\cO(1/\epsilon)$ iteration complexity for finding an $\epsilon$-optimal 
solution. For their second method, each iterate requires solving a box 
constrained quadratic programming, and it has a local linear 
convergence rate. However, its global iteration complexity for finding an 
$\epsilon$-optimal solution is theoretically unknown. After the first release 
of our paper, Friedman et al.\ \cite{FriHasTib07} studied a slight variant 
of the BCD method proposed in \cite{DaOnEl06}. At each iteration of their method, 
a coordinate descent approach is applied to solve a lasso ($l_1$-regularized) 
least-squares problem, which is the dual of the box constrained quadratic 
programming appearing in the BCD method \cite{DaOnEl06}. In contrast with these 
methods, the smooth optimization approach proposed in this paper has a more 
attractive iteration complexity that is $\cO(1/{\sqrt{\epsilon}})$ for finding 
an $\epsilon$-optimal solution. In addition, we propose a variant of the smooth 
optimization approach which has substantially outperformed the latter one in our 
computational experiments. We also compare the performance of our approaches 
with their methods for sparse covariance selection on a set of randomly generated 
instances. It shows that our smooth optimization approach substantially outperforms 
their first method above (i.e., Nesterov's smooth approximation scheme), and 
moreover, its variant substantially outperforms their methods 
\cite{DaOnEl06,FriHasTib07} mentioned above.    
                    
The paper is organized as follows. In Section \ref{smooth-appr}, we 
introduce a class of non-smooth concave maximization problems in which 
we are interested, and propose a smooth optimization approach to them. 
In Section \ref{cov-select}, we briefly introduce sparse covariance selection, 
and show that it can be approximately solved as a $l_1$-norm penalized 
maximum likelihood estimation problem. We also discuss the application of the 
smooth optimization approach for solving this problem, and propose a variant 
of this approach. In Section \ref{comp}, we compare the performance of our 
smooth optimization approach and its variant with two other first-order methods 
studied in \cite{DaOnEl06,FriHasTib07} for sparse covariance selection on a 
set of randomly generated instances. Finally, we present some concluding 
remarks in Section \ref{concl-remark}.

\subsection{Notation} \label{notation}
In this paper, all vector spaces are assumed to be finite dimensional. 
The space of symmetric $n \times n$ matrices will be denoted by $\cS^n$. 
If $X \in \cS^n$ is positive semidefinite, we write $X \succeq 0$. Also, we 
write $X \preceq Y$ to mean $Y-X \succeq 0$. The cone of positive semidefinite 
(resp., definite) matrices is denoted by $\cS^n_+$ (resp., $\cS^n_{++}$). Given
matrices $X$ and $Y$ in $\Re^{p \times q}$, the standard inner product is 
defined by $\langle X, Y \rangle := \tr (XY^T)$, where $\tr(\cdot)$ denotes 
the trace of a matrix. $\|\cdot\|$ denotes the Euclidean norm and its associated 
operator norm unless it is explicitly stated otherwise. The Frobenius norm of a 
real matrix $X$ is defined as $\|X\|_F := \sqrt{\tr(XX^T)}$. We denote by $e$ 
the vector of all ones, and by $I$ the identity matrix. Their dimensions should 
be clear from the context. For a real matrix $X$, we denote by $\card(X)$ the 
cardinality of $X$, that is, the number of nonzero entries of $X$, and denote by 
$|X|$ the absolute value of $X$, that is, $|X|_{ij}=|X_{ij}|$ for all $i,j$. The 
determinant and the minimal (resp., maximal) eigenvalue of a real symmetric matrix 
$X$ are denoted by $\det X$ and $\lambda_{\min}(X)$ (resp., $\lambda_{\max}(X)$), 
respectively. For a $n$-dimensional vector $w$, $\diag(w)$ denote the diagonal 
matrix whose $i$-th diagonal element is $w_i$ for $i=1,\ldots,n$. We denote by 
$\cZ_{+}$ the set of all nonnegative integers.      

Let the space $\mF$ be endowed with an arbitrary norm $\|\cdot\|$. The dual space 
of $\mF$, denoted by $\mF^*$, is the normed real vector space consisting of all 
linear functionals of $s: \mF \to \Re$, endowed with the dual norm $\|\cdot\|^*$ 
defined as 
\[
\|s\|^* := \max\limits_u \{\langle s, u \rangle: \ \|u\| \le 1 \},
\ \ \ \forall s\in \mF^*,
\]
where $\langle s,u\rangle := s(u)$ is the value of the linear functional 
$s$ at $u$. Finally, given an operator $\cA: \mF \to \mF^*$, we define 
\[
\cA [H, H] := \langle \cA H, H\rangle 
\]
for any $H \in \mF$.

\section{Smooth optimization approach}
\label{smooth-appr}

%Let $X$ and $U$ be nonempty convex compact sets in finite-dimensional 
%real vector spaces $\mE$ and $\mF$, respectively. Let $\phi(x,u): X \times U 
%\to \Re $ be a continuous function which is strictly concave in $x\in X$ for 
%every fixed $u\in U$, and convex differentiable in $u\in U$ for every fixed $x\in X$.
In this section, we consider a class of concave non-smooth maximization problems: 
\beq \label{concave-opt}
\max\limits_{x\in X} \, \{ g(x) := \min\limits_{u\in U} \phi(x,u)\}, 
\eeq  
where $X$ and $U$ are nonempty convex compact sets in finite-dimensional 
real vector spaces $\mE$ and $\mF$, respectively, and $\phi(x,u): X \times U 
\to \Re $ is a continuous function which is {\it strictly} concave in $x\in X$ for 
every fixed $u\in U$, and convex differentiable in $u\in U$ for every fixed 
$x\in X$. Therefore, for any $u\in U$, the function 
\beq \label{fu}
f(u) := \max\limits_{x\in X}\phi(x,u) 
\eeq 
is well-defined. We also easily conclude that $f(u)$ is 
convex differentiable on $U$, and its gradient is given by 
\beq \label{gu}
\nabla f(u) = \nabla_u \phi(x(u),u), \ \ \forall u\in U,
\eeq
where $x(u)$ denotes the unique solution of \eqnok{fu}.  

Let the space $\mF$ be endowed with an arbitrary norm $\|\cdot\|$. We further 
assume that $\nabla f(u)$ is Lipschitz continuous on $U$ with respect to $\|\cdot\|$, 
i.e., there exists some $L>0$ such that 
\[
\|\nabla f(u) - \nabla f(\tilde u)\|^* \le L \|u - \tilde u\|,   
\ \ \ \forall u, \tilde u \in U.
\]

Under the above assumptions, we easily observe that: i) problem \eqnok{concave-opt} 
and its dual, that is, 
\beq \label{dual-prob}
\min\limits_u \{f(u): \ u \in U\},
\eeq
are both solvable and have the same optimal value; and ii) the dual problem 
\eqnok{dual-prob} can be suitably solved by Nesterov's smooth minimization 
approach \cite{Nest83-1, Nest05-1}. 

Denote by $d(u)$ a prox-function of the set $U$. We assume that $d(u)$ is 
continuous and strongly convex on $U$ with modulus $\sigma > 0$. 
Let $u_0$ be the center of the set $U$ defined as 
\beq \label{u0}
u_0 = \arg\min\{d(u): \ u \in U\}.
\eeq
Without loss of generality assume that $d(u_0)=0$. We now describe Nesterov's 
smooth minimization approach \cite{Nest83-1, Nest05-1} for solving the dual 
problem \eqnok{dual-prob}, and we will show that it simultaneously solves 
the non-smooth concave maximization problem \eqnok{concave-opt}.
%non-smooth concave maximization problem \eqnok{concave-opt}. 

\gap

\noindent
\begin{minipage}[h]{6.6 in}
{\bf Smooth Minimization Algorithm:} \\ [5pt]
Let $u_0 \in U$ be given in \eqnok{u0}. Set $x_{-1}=0$ and $k=0$.
\begin{itemize}
\item[1)]
Compute $\nabla f(u_k)$ and $x(u_k)$. Set $x_{k} = \frac{k}{k+2} x_{k-1} 
+ \frac{2}{k+2} x(u_k)$.
\item[2)]
Find $u^{sd}_{k} \in \Argmin \left \{ \langle \nabla f(u_k), u-u_k \rangle +
\frac{L}2 \, \|u-u_k\|^2: \ u \in U \right \}$.
\item[3)]
Find $u^{ag}_{k} = \argmin \left \{ \frac{L}{\sigma}d(u)+\sum\limits_{i=0}^k
\frac{i+1}2[f(u_i) + \langle \nabla f(u_i), u-u_i \rangle]: \ u \in U \right \}$. 
\item[4)]
Set $u_{k+1} = \frac{2}{k+3} u^{ag}_{k} + \frac{k+1}{k+3} u^{sd}_{k}$.
\item[5)]
Set $k \leftarrow k+1$ and go to step 1). 
%until $f(u^{sd}_k) - g(x_k) \le \epsilon$.
\end{itemize}
\noindent
{\bf end}
\end{minipage}
\vgap

The following property of the above algorithm is established in Theorem 2 
of Nesterov \cite{Nest05-1}.

\begin{theorem}\label{thm1}
Let the sequence $\{(u_k, u^{sd}_k) \}^{\infty}_{k=0} \subseteq U \times U$ 
be generated by the Smooth Minimization Algorithm. Then for any $k \ge 0$ 
we have 
\beq \label{p-conv}
\frac{(k+1)(k+2)}4 f(u^{sd}_k) \le \min\left\{\frac{L}{\sigma}d(u)+
\sum\limits_{i=0}^k \frac{i+1}2 [f(u_i)+\langle \nabla f(u_i), u-u_i \rangle]: 
\ u \in U \right\}.
\eeq  

\end{theorem}

We are ready to establish the main convergence result of the Smooth Minimization 
Algorithm for solving the non-smooth concave maximization problem \eqnok{concave-opt} 
and its dual \eqnok{dual-prob}. Its proof is a generalization of the one 
given in a more special context in \cite{Nest05-1}.

\begin{theorem} \label{mtm-concave}
After $k$ iterations, the Smooth Minimization Algorithm generates a pair of approximate 
solutions $(u^{sd}_k, x_k)$ to problem \eqnok{concave-opt} and its dual \eqnok{dual-prob}, 
respectively, which satisfy the following inequality:
\beq \label{converg-ineq}
0 \le f(u^{sd}_k) - g(x_k) \le \frac{4LD}{\sigma(k+1)(k+2)}. 
\eeq 
Thus if the termination criterion $f(u^{sd}_k) - g(x_k) \le \epsilon$ is applied, 
the iteration complexity of finding an $\epsilon$-optimal solution to problem 
\eqnok{concave-opt} and its dual \eqnok{dual-prob} by the Smooth Minimization 
Algorithm does not exceed $2\sqrt{\frac{LD}{\sigma \epsilon}}$, where 
\beq \label{D}
D = \max \{d(u): \ u \in U\}.
\eeq
\end{theorem}

\begin{proof}
In view of \eqnok{fu}, \eqnok{gu} and the notation $x(u)$, we have 
\beq \label{eq1}
f(u_i) + \langle \nabla f(u_i), u-u_i \rangle = \phi(x(u_i),u_i) + 
\langle \nabla_u \phi(x(u_i),u_i), u-u_i \rangle.
\eeq
Invoking the fact that the function $\phi(x, \cdot)$ is convex on $U$ for 
every fixed $x\in X$, we obtain  
\beq \label{convex-ineq}
\phi(x(u_i),u_i) + \langle \nabla_u \phi(x(u_i),u_i), u-u_i \rangle \le 
\phi(x(u_i),u).
\eeq 
Notice that $x_{-1}=0$, and $x_{k} = \frac{k}{k+2} x_{k-1} + 
\frac{2}{k+2} x(u_k)$ for any $k \ge 0$, which imply      
\beq \label{xk}
x_k = \sum\limits_{i=0}^k \frac{2(i+1)}{(k+1)(k+2)}x(u_i).
\eeq
Using \eqnok{eq1}, \eqnok{convex-ineq}, \eqnok{xk} and the fact that the 
function $\phi(\cdot,u)$ is concave on $X$ for every fixed $u\in U$, we have
\beqas
\sum\limits_{i=0}^k (i+1)[f(u_i) + \langle \nabla f(u_i), u-u_i \rangle ] 
& \le & \sum\limits_{i=0}^k (i+1) \phi(x(u_i),u) \\ [5pt]
& \le & \frac12(k+1)(k+2)\phi(x_k,u)
\eeqas 
for all $u\in U$. It follows from this relation, \eqnok{p-conv}, \eqnok{D} 
and \eqnok{concave-opt} that 
\beqas 
f(u^{sd}_k) & \le & \frac{4LD}{\sigma(k+1)(k+2)} + \min\limits_u 
\left\{\sum\limits_{i=0}^k \frac{2(i+1)}{(k+1)(k+2)} [f(u_i)+\langle 
\nabla f(u_i), u-u_i \rangle]: \ u \in U \right\} \\ [5pt]
& \le &  \frac{4LD}{\sigma(k+1)(k+2)} + \min_{u \in U} \phi(x_k, u) 
= \frac{4LD}{\sigma(k+1)(k+2)} + g(x_k),
\eeqas
and hence the inequality \eqnok{converg-ineq} holds. The remaining 
conclusion directly follows from \eqnok{converg-ineq}.
\end{proof}

\vgap

\begin{remark}
We shall mention that Nesterov \cite{Nest05-2} developed the excessive gap 
technique for solving problem \eqnok{concave-opt} and its dual \eqnok{dual-prob} 
in a special context, which enjoys the same iteration complexity as the 
Smooth Minimization Algorithm described above. In addition, it is not hard 
to observe that the technique proposed in \cite{Nest05-2} can be extended
to solve problem \eqnok{concave-opt} and its dual \eqnok{dual-prob} in 
the aforementioned general framework provided that the subproblem 
\beq \label{prox-subprob}
\min\limits_{u \in U} \, \phi(x,u) + \mu d(u)
\eeq
can be suitably solved for any given $\mu>0$ and $x\in X$. The computation 
of each iterate of Nesterov's excessive gap technique \cite{Nest05-2} is 
similar to that of the Smooth Minimization Algorithm except that the 
former method requires solving a prox subproblem in the form of 
\eqnok{prox-subprob}, but the latter one needs to solve the prox subproblem 
described in step 3) above. When the function $\phi(x,\cdot)$ is affine for 
every fixed $x\in X$, these two prox subproblems have the same form, and thus 
the computational cost of Nesterov's excessive gap technique 
\cite{Nest05-2} is almost same as that of the Smooth Minimization Algorithm;
however, for a more general function $\phi(\cdot,\cdot)$, the computational 
cost of the former method can be more expensive than that of 
the latter method. 
\end{remark}

\vgap

The following results will be used to develop a variant of the Smooth 
Minimization Algorithm for sparse covariance selection in Subsection 
\ref{variant}.

\begin{lemma} \label{unique}
Problem \eqnok{concave-opt} has a unique optimal solution, denoted by $x^*$. 
Moreover, for any $u^* \in \Argmin\{f(u):\ u\in U\}$, we have 
\beq \label{opt-rel}
x^* = \arg\max\limits_{x\in X} \phi(x, u^*).
\eeq  
\end{lemma}

\begin{proof}
%Since the function $\phi(x,u)$ is continuous on the compact set $X \times U$, we 
%easily know that $g(x)$ defined in \eqnok{concave-opt} is continuous on the 
%compact set $X$. Thus, 
We clearly know that problem \eqnok{concave-opt} has an optimal solution. 
To prove its uniqueness, it suffices to show that $g(x)$ is strictly 
concave on $X$. Indeed, since $X \times U$ is a convex compact set and 
$\phi(x, u)$ is continuous on $X \times U$, it follows that for any 
$t\in (0,1)$, $x^1 \neq x^2 \in X$, there exists some $\tu \in U$ such 
that   
\[
\phi(tx^1+(1-t)x^2, \tu) = \min_{u\in U}\phi(tx^1+(1-t)x^2, u).
\] 
Recall that $\phi(\cdot,u)$ is strictly concave on $X$ for every fixed $u\in U$. 
Therefore, we have 
\[
\ba{lcl}
\phi(tx^1+(1-t)x^2, \tu) & > & t \phi(x^1, \tu) + (1-t) \phi(x^2, \tu), \\ [4pt]
& \ge & t \min\limits_{u\in U}\phi(x^1, u) + (1-t) \min\limits_{u\in U}\phi(x^2, u),
\ea
\]
which together with \eqnok{concave-opt} implies that 
\[
g(tx^1+(1-t)x^2) > t g(x^1) + (1-t) g(x^2)
\]
for any $t\in (0,1)$, $x^1 \neq x^2 \in X$, and hence, $g(x)$ is strictly 
concave on $X$ as desired. 

Note that $x^*$ is the optimal solution of problem \eqnok{concave-opt}. We clearly 
know that for any $u^* \in \Argmin \{f(u):\ u\in U\}$, $(u^*, x^*)$ is a saddle point 
for problem \eqnok{concave-opt}, that is, 
\[
\phi(x^*, u) \ge \phi(x^*, u^*) \ge \phi(x, u^*), \ \ \ \forall (x,u) \in X \times U,
\]
and hence, we have 
\[
x^* \in \Arg\max\limits_{x\in X} \phi(x, u^*).
\]
It together with the fact that $\phi(\cdot,u^*)$ is strictly concave on $X$, immediately 
yields \eqnok{opt-rel}.  
\end{proof} 

\vgap

\begin{theorem} \label{prim-conv}
Let $x^*$ be the unique optimal solution of \eqnok{concave-opt}, and $f^*$ be the optimal value 
of problems \eqnok{concave-opt} and \eqnok{dual-prob}. Assume that the sequences $\{u_k\}^{\infty}_{k=0}$ 
and $\{x(u_k)\}^{\infty}_{k=0}$ are generated by the Smooth Minimization Algorithm. Then the 
following statements hold:
\bi
\item[1)] $f(u_k) \to f^*$, $x(u_k) \to x^*$ as $k \to \infty$;
\item[2)] $f(u_k) - g(x(u_k)) \to 0$ as $k \to \infty$.
\ei  
\end{theorem}  

\begin{proof}
Recall from the Smooth Minimization Algorithm that 
\[
u_{k+1} =  \left(2 u^{ag}_{k}+(k+1) u^{sd}_{k}\right)/(k+3),
\ \ \forall k \ge 0.
\]  
Since $u^{sd}_{k}, u^{ag}_{k} \in U$ for $\forall k\ge 0$, and $U$ is a compact 
set, we have $u_{k+1}-u^{sd}_{k} \to 0$ as $k \to \infty$. Notice that $f(u)$ is 
continuous on the compact set $U$, and hence, it is uniformly continuous on $U$. 
Then we further have $f(u_{k+1})-f(u^{sd}_{k}) \to 0$ as 
$k \to \infty$. Also, it follows from Theorem \ref{mtm-concave} that 
$f(u^{sd}_{k}) \to f^*$ as $k \to \infty$. Therefore, we conclude that $f(u_k) \to f^*$ 
as $k \to \infty$. 

Note that $X$ is a compact set, and $x(u_k) \subseteq X$ for $\forall k \ge 0$.  
To prove that $x(u_k) \to x^*$ as $k \to \infty$, it suffices to show that every convergent 
subsequence of $\{x(u_k)\}^{\infty}_{k=0}$ converges to $x^*$ as $k \to \infty$. Indeed, 
assume that $\{x(u_{n_k})\}^{\infty}_{k=0}$ is an arbitrary convergent subsequence, and 
$x(u_{n_k}) \to \tx^*$ as $k \to \infty$ for some $\tx^*\in X$. Without 
loss generality, assume that the sequence $\{u_{n_k}\}^{\infty}_{k=0} \to \tu^*$ 
as $k \to \infty$ for some $\tu^* \in U$ (otherwise, one can consider any convergent 
subsequence of $\{u_{n_k}\}^{\infty}_{k=0}$). Using the result that $f(u_k) \to f^*$, 
we obtain that 
\[
\phi\left(x(u_{n_k}), u_{n_k}\right) = f(u_{n_k}) \to f^*, \ \ \ \mbox{as} \ k \to \infty.   
\] 
Upon letting $k \to \infty$ and using the continuity of $\phi(\cdot, \cdot)$, we have 
$\phi(\tx^*,\tu^*) = f(\tu^*) = f^*$. Hence, it follows that
\[ 
\tu^* \in \Arg\min\limits_{u\in U} f(u), \ \ \ \tx^* =\arg\max\limits_{x\in X} \phi(x, \tu^*),
\]
which together with Lemma \ref{unique} implies that $\tx^*=x^*$. Hence as desired,  
$x(u_{n_k}) \to x^*$ as $k \to \infty$. 

%Since $\phi(\cdot, \cdot)$ is continuous over compact set $X \times U$, we clearly see 
As shown in Lemma \ref{unique}, the function $g(x)$ is continuous on $X$. This result 
together with statement 1) immediately implies that statement 2) holds. 
\end{proof}

\section{Sparse covariance selection}
\label{cov-select}

In this section, we discuss the application of the smooth optimization approach 
proposed in Section \ref{smooth-appr} to sparse covariance selection. More 
specifically, we briefly introduce sparse covariance selection in Subsection 
\ref{intro-cov}, and show that it can be approximately solved as a $l_1$-norm 
penalized maximum likelihood estimation problem in Subsection \ref{reform}. In 
Subsection \ref{appl}, We address some implementation details of the smooth 
optimization approach for solving this problem, and propose a variant of this 
approach in Subsection \ref{variant}.
%we 
%and its variant for solving it.

\subsection{Introduction of sparse covariance selection}
\label{intro-cov}

In this subsection, we briefly introduce sparse covariance selection. For more 
details, see d'Aspremont et al.\ \cite{DaOnEl06} and the references therein.
 
Given $n$ variables with a Gaussian distribution $\cN(0, C)$ for which the true 
covariance matrix $C$ is unknown, we are interested in estimating $C$ from a 
sample covariance matrix $\Sigma$ by maximizing its likelihood while 
imposing a certain number of components in the inverse of the estimation of 
$C$ to zero. This problem is commonly known as {\it sparse covariance selection} 
(see \cite{Demp72}). Since zeros in the inverse of covariance matrix correspond 
to conditional independence in the model, sparse covariance selection can be used 
to determine a robust estimate of the covariance matrix, and simultaneously 
discover the sparse structure in the underlying graphical model. 

Several approaches have been proposed for sparse covariance selection in literature. 
For example, Bilmes \cite{Bilmes00} proposed a method based on choosing statistical 
dependencies according to conditional mutual information computed from training 
data. The recent works \cite{JoCaDo04,DoWe04} involve identifying the Gaussian graphical 
models that are best supported by the data and any available prior information 
on the covariance matrix. 
%In order to trade-off the likelihood of the estimation of covariance with 
%the sparsity of its inverse, 
Given a sample covariance matrix $\Sigma \in \cS^n_+$, d'Aspremont et al.\ 
\cite{DaOnEl06} recently formulated sparse covariance selection as the following 
estimation problem:
\beq \label{card-prob}
\begin{array}{ll}
\max\limits_X & \log\det X - \langle \Sigma, X \rangle - \rho \card(X)\\ 
\mbox{s.t.} & \talpha I \preceq X \preceq \tbeta I,
\end{array}   
\eeq 
where $\rho>0$ is a parameter controlling the trade-off between likelihood and 
cardinality, and $0 \le \talpha < \tbeta \le \infty$ are the fixed bounds on the 
eigenvalues of the solution. For some specific choices of $\rho$, 
the formulation \eqnok{card-prob} has been used for model selection in 
\cite{Akaike73,BuAn04}, and applied to speech recognition and gene network 
analysis (see \cite{Bilmes00, DoHa04}). 

Note that the estimation problem \eqnok{card-prob} itself is a NP-hard combinatorial 
problem because of the penalty term $\card(X)$. To overcome the computational 
difficulty, d'Aspremont et al.\ \cite{DaOnEl06} used an argument that is often 
used in regression techniques (e.g., see \cite{Tibs96,ChDoSa98,DonTan05}), where 
sparsity of the solution is concerned, to relax $\card(X)$ to $e^T |X| e$, and 
obtained the following $l_1$-norm penalized maximum likelihood estimation 
problem:
\beq \label{relax}
\begin{array}{ll}
\max\limits_X & \log\det X - \langle \Sigma, X \rangle - \rho e^T |X| e\\ 
\mbox{s.t.} & \talpha I \preceq X \preceq \tbeta I,
\end{array}   
\eeq
Recently, Yuan and Lin \cite{YuLi07-1} proposed a similar estimation problem for 
sparse covariance selection given as follows:
\beq \label{relax-yl}
\begin{array}{ll}
\max\limits_X & \log\det X - \langle \Sigma, X \rangle - \rho \sum\limits_{i\neq j} |X_{ij}|\\ 
\mbox{s.t.} & \talpha I \preceq X \preceq \tbeta I,
\end{array}   
\eeq
with $\talpha=0$ and $\tbeta=\infty$. They showed that problem \eqnok{relax-yl} can be 
suitably solved by the interior point algorithm developed in Vandenberghe et al.\ \cite{VaBoWu98}. 
A few other approaches have also been studied for sparse covariance selection by solving some 
related maximum likelihood estimation problems in literature. For example, Huang 
et al.\ \cite{HuLiPo06} proposed an iterative (heuristic) algorithm to minimize a nonconvex 
penalized likelihood. Dahl et al.\ \cite{DaVaRo06,DaRoVa04} applied Newton method, 
coordinate steepest descent method, and conjugate gradient method for the problems 
for which the conditional independence structure is partially known.

As shown in d'Aspremont et al.\ \cite{DaOnEl06} (see also \cite{BaElDaNa06}), and 
Yuan and Lin \cite{YuLi07-1}, the $l_1$-norm penalized maximum likelihood 
estimation problems \eqnok{relax} and \eqnok{relax-yl} are capable of discovering 
effectively the sparse structure, or equivalently, the conditional independence in 
the underlying graphical model. Also, it is not hard to see that the estimation 
problem \eqnok{relax-yl} becomes a special case of problem \eqnok{relax} if 
replacing $\Sigma$ by $\Sigma + \rho I$ in \eqnok{relax-yl}. For these reasons, 
we focus on problem \eqnok{relax} only for the remaining paper.
 
\subsection{Non-smooth strictly concave maximization reformulation}
\label{reform}

In this subsection, we show that problem \eqnok{relax} can be reformulated as a 
non-smooth strictly concave maximization problem of the form \eqnok{concave-opt}. 

Recall from Subsection \ref{intro-cov} that $\Sigma \in \cS^n_{+}$, and keep in mind 
that the notations $|\cdot|$, $\|\cdot\|$ and $\|\cdot\|_F$ are defined in 
Subsection \ref{notation}. We first provide some tighter bounds on the optimal 
solution of problem \eqnok{relax} for the case where $\talpha=0$ and $\tbeta=\infty$. 

\begin{proposition} \label{bound}
Assume that $\talpha=0$ and $\tbeta=\infty$. Let $X^* \in \cS^n_{++}$ be the unique optimal 
solution of problem \eqnok{relax}. Then we have $\alpha I \preceq X^* \preceq \beta I$, where 
\beq \label{alpha}
\alpha = \frac{1}{\|\Sigma\|+n\rho}, \ \ \ 
\beta = \min\left\{\frac{n-\alpha \tr(\Sigma)}{\rho}, \eta \right\}
\eeq
with
\[
\eta = \left\{\ba{ll}
\min\left\{e^T |\Sigma^{-1}|e, (n - \rho \sqrt{n} \alpha)\|\Sigma^{-1}\| - (n-1) \alpha\right\},  
& \ \mbox{if} \ \Sigma \ \mbox{is invertible}; \\ [8pt]
2 e^T |(\Sigma+\frac{\rho}{2}I)^{-1}|e - \tr((\Sigma+\frac{\rho}{2}I)^{-1}), & \ \mbox{otherwise}. 
\ea \right.
\]  
\end{proposition}

\begin{proof}
Let
\beq \label{cU}
\cU := \{U\in \cS^n: \ |U_{ij}| \le 1, \forall ij\},  
\eeq
and 
\beq \label{lagr}
L(X, U) = \log\det X - \langle \Sigma+\rho U, X \rangle, \ \ \ \forall (X,U) \in \cS^n_{++} \times \cU. 
\eeq  
Note that $X^* \in \cS^n_{++}$ is the optimal solution of problem \eqnok{relax}.
It can be easily shown that there exists some $U^*\in \cU$ such that $(X^*, U^*)$ is 
a saddle point of $L(\cdot, \cdot)$ on $\cS^n_{++} \times \cU$, that is, 
\[
X^* = \arg\min\limits_{X \in \cS^n_{++}} L(X,U^*), \ \ \ \ U^* \in \Arg\min\limits_{U\in \cU} L(X^*,U).  
\]
The above relations along with \eqnok{cU} and \eqnok{lagr} immediately yield
\beq \label{optcond}
X^*(\Sigma + \rho U^*) = I, \ \ \ \ \langle X^*, U^* \rangle = e^T |X^*|e.
\eeq
Hence, we have   
\[
X^* = (\Sigma + \rho U^*)^{-1} \succeq \frac{1}{\|\Sigma\|+ \rho \|U^*\|} I,
\]
which together with \eqnok{cU} and the fact $U^*\in \cU$, implies that 
$X^* \succeq \frac{1}{\|\Sigma\|+n\rho} I$. Thus as desired, $X^* \succeq \alpha I$, 
where $\alpha$ is given in \eqnok{alpha}. 

We next bound $X^*$ from above. In view of \eqnok{optcond}, we have
\beq \label{trace}
\langle X^*, \Sigma \rangle + \rho e^T |X^*|e = n,
\eeq
which together with the relation $X^* \succeq \alpha I$ implies that 
\beq \label{eXe1}
e^T |X^*|e \le \frac{n-\alpha \tr(\Sigma)}{\rho}.
\eeq
Now let $X(t) := (\Sigma + t \rho I)^{-1}$ for $t\in (0,1)$. By concavity 
of $\log\det(\cdot)$, one can easily see that $X(t)$ maximizes the function 
$\log\det(\cdot) - \langle \Sigma + t \rho I, \cdot \rangle$ over 
$\cS^n_{++}$. Using this observation and the definition of $X^*$, we can have 
\beqas
\log\det X^* - \langle \Sigma + t \rho I, X^* \rangle &\le & \log\det X(t) - 
\langle \Sigma + t \rho I, X(t) \rangle, \\ [5pt]
\log\det X(t) - \langle \Sigma, X(t) \rangle - \rho e^T |X(t)| e & \le& 
\log\det X^* - \langle \Sigma, X^* \rangle - \rho e^T |X^*| e.
\eeqas 
Adding the above two inequalities upon some algebraic simplification, we 
obtain that
\[
e^T |X^*| e  - t \tr(X^*) \le e^T |X(t)| e - t \tr(X(t)), 
\]
and hence,
\beq \label{exe-bdd}
e^T |X^*| e \le \frac{e^T|X(t)|e - t\tr(X(t))}{1-t}, \ \ \forall t\in (0,1).
\eeq 
If $\Sigma$ is invertible, upon letting $t \downarrow 0$ on both sides of \eqnok{exe-bdd},
we have
\[
e^T |X^*| e \le e^T |\Sigma^{-1}|e.
\]
Otherwise, letting $t=1/2$ in \eqnok{exe-bdd}, we obtain 
\[
e^T |X^*| e \le 2 e^T |(\Sigma+\frac{\rho}{2}I)^{-1}|e - \tr((\Sigma+\frac{\rho}{2}I)^{-1}).
\]
Combining the above two inequalities and \eqnok{eXe1}, we have  
\beq \label{eXe2}
\|X^*\| \le \|X^*\|_F \le e^T |X^*| e \le \min\left\{\frac{n-\alpha \tr(\Sigma)}{\rho}, 
\gamma \right\},
\eeq 
where 
\[
\gamma = \left\{\ba{ll}
e^T |\Sigma^{-1}|e,  & \ \mbox{if} \ \Sigma\ \mbox{is invertible}; \\ [5pt]
2 e^T |(\Sigma+\frac{\rho}{2}I)^{-1}|e - \tr((\Sigma+\frac{\rho}{2}I)^{-1}), & 
\ \mbox{otherwise}. 
\ea \right.
\]
Further, using the relation $X^* \succeq \alpha I$, we obtain that 
\[
e^T |X^*| e \ge \|X^*\|_F \ge \sqrt{n} \alpha,
\]
which together with \eqnok{trace} implies that 
\[
\tr(X^*\Sigma) \le n - \rho \sqrt{n} \alpha.
\] 
This inequality along with the relation $X^* \succeq \alpha I$ yields  
\[
\lambda_{\min}(\Sigma) ((n-1)\alpha + \|X^*\|) \le \tr(X^*\Sigma) \le 
n - \rho \sqrt{n} \alpha.
\]
Hence if $\Sigma$ is invertible, we further have  
\[
\|X^*\| \le (n - \rho \sqrt{n} \alpha)\|\Sigma^{-1}\| - (n-1) \alpha.
\] 
This together with \eqnok{eXe2} implies that $X^* \preceq \beta I$, 
where $\beta$ is given in \eqnok{alpha}.
\end{proof}

\vgap

\begin{remark}
Some bounds on $X^*$ were also derived in d'Aspremont et al.\ \cite{DaOnEl06}. 
In contrast with their bounds, our bounds given in \eqnok{alpha} are tighter. 
Moreover, our approach for deriving the above bounds can be generalized to 
handle the case where $\talpha >0$ and $\tbeta=\infty$, but their approach 
cannot. Indeed, if $\talpha >0$ and $\tbeta=\infty$, we can set 
$\alpha = \talpha$, and replace the above $X(t)$ by the optimal 
solution of 
\[
\begin{array}{ll}
\max\limits_X  & \log\det X - \langle \Sigma+\rho I, X \rangle \\ 
\mbox{s.t.} & \talpha I \preceq X,
\end{array} 
\] 
which has a closed-form expression. By following a similar derivation as 
above, one can obtain a positive scalar $\beta$ such that $X^* \preceq \beta I$. In 
addition, for the case where $\talpha=0$ and $0 <\tbeta <\infty$, one can 
set $\beta = \tbeta$ and easily show that $X^* \ge \alpha I$, where 
$\alpha = \beta e^{-\beta(\tr(\Sigma)+n\rho)}$.   
\end{remark}

\vgap

From the above discussion, we conclude that problem \eqnok{relax} is equivalent to 
the following problem:
\beq \label{new-relax}
\begin{array}{ll}
\max\limits_X  & \log\det X - \langle \Sigma, X \rangle - \rho e^T |X| e\\ 
\mbox{s.t.} & \alpha I \preceq X \preceq \beta I,
\end{array}   
\eeq
for some $0 < \alpha < \beta < \infty$.

We further observe that problem \eqnok{new-relax} can be rewritten as  
\beq \label{smooth-saddle}
\max\limits_{X \in \cX} \min\limits_{U \in \cU} \ 
\log\det X - \langle \Sigma + \rho U, X \rangle, 
\eeq 
where $\cU$ is defined in \eqnok{cU}, and $\cX$ is defined as follows:
\beq \label{cX}
\cX := \{X\in\cS^n: \ \alpha I \preceq X \preceq \beta I\}.
\eeq
Therefore, we conclude that problem \eqnok{relax} is equivalent to 
\eqnok{smooth-saddle}. For the remaining paper, we will focus on 
problem \eqnok{smooth-saddle} only. 

\subsection{Smooth optimization method for sparse covariance selection}
\label{appl} 

In this subsection, we describe the implementation details of the Smooth 
Minimization Algorithm proposed in Section \ref{smooth-appr} for solving 
problem \eqnok{smooth-saddle}. We also compare the complexity of this 
algorithm with interior point methods, and two other first-order methods 
studied in d'Aspremont et al.\ \cite{DaOnEl06}, that is, Nesterov's smooth 
approximation scheme and block coordinate descent method. 

We first observe that the sets $\cX$ and $\cU$ both lie in the space $\cS^n$, 
where $\cX$ and $\cU$ are defined in \eqnok{cX} and \eqnok{cU}, respectively. 
Let $\cS^n$ be endowed with the Frobenius norm, and let $\td(X) = \log\det X$ 
for $X\in\cX$. Then for any $X \in \cX$, we have
\[
\nabla^2\td(X)[H, H] = -\tr(X^{-1}HX^{-1}H) \le -\beta^{-2} \|H\|^2_F
\]
for all $H\in \cS^n$, and hence, $\td(X)$ is strongly concave on $\cX$ 
with modulus $\beta^{-2}$. Using this result and Theorem 1 
of \cite{Nest05-1}, we immediately conclude that $\nabla f(U)$ is Lipschitz 
continuous with constant $L = \rho^2\beta^2$ on $\cU$, where 
\beq \label{fU}
f(U) := \max\limits_{X \in \cX} \ \log\det X - \langle \Sigma + \rho U, 
X \rangle, \ \ \ \forall U\in\cU.  
\eeq
Denote the unique optimal solution of problem \eqnok{fU} by $X(U)$. 
For any $U\in\cU$, we can compute $X(U)$, $f(U)$ and $\nabla f(U)$ as 
follows. 

Let $\Sigma + \rho U = Q \diag(\gamma) Q^T$ be an eigenvalue 
decomposition of $\Sigma + \rho U$ such that $QQ^T=I$. For $i=1,\ldots,n$, 
let 
\[
\lambda_i = \left\{\ba{ll} \min\{\max\{1/{\gamma_i}, \alpha\}, \beta\}, 
& \ \mbox{if} \ \gamma_i > 0; \\ 
\beta, \ & \ \mbox{otherwise}.
\ea \right. 
\] 
It is not hard to show that 
\beq \label{xfu}
X(U)=Q\diag(\lambda)Q^T, \ \ \ 
f(U) = -\gamma^T\lambda + \sum\limits_{i=1}^n \log\lambda_i, 
\ \ \ \nabla f(U) = -\rho X(U).
\eeq

From the above discussion, we see that problem \eqnok{smooth-saddle} has 
exactly the same form as \eqnok{concave-opt}, and also satisfies all assumptions 
imposed on problem \eqnok{concave-opt}. Therefore, it can be suitably solved by the 
Smooth Minimization Algorithm proposed in Section \ref{smooth-appr}. The 
implementation details of this algorithm for problem \eqnok{smooth-saddle} are 
described as follows. 

Given $U_0\in\cU$, let $d(U) = \|U-U_0\|^2_F/2$ be the proximal function on 
$\cU$, which is strongly convex function with modulus $\sigma=1$. 
%By \eqnok{cU} and \eqnok{D}, we have $D=n^2/2$. According 
%to \eqnok{u0}, we shall choose $U_0=0$ for such $d(U)$. 
For our specific choice of the norm and $d(U)$, we clearly see that steps 2) 
and 3) of the Smooth Minimization Algorithm can be solved as a problem of the 
form
\[
V = \arg\min\limits_{U\in \cU} \langle G, U \rangle + \|U\|^2_F/2
\]
for some $G\in\cS^n$. In view of \eqnok{cU}, we see that 
\[
V_{ij} = \max\{\min\{-G_{ij},1\}, -1\}, \ i,j=1,\ldots,n.
\]  
In addition, for any $X\in\cX$, we define
\beq \label{gx}  
g(X) := \log\det X - \langle \Sigma, X \rangle - \rho e^T |X| e.
\eeq

For the ease of comparison with its latter variant, we now present a 
complete version of the aforementioned Smooth Minimization Algorithm 
for solving problem \eqnok{smooth-saddle} and its dual.

\gap

\noindent
\begin{minipage}[h]{6.6 in}
{\bf Smooth Minimization Algorithm for Covariance Selection (SMACS):} \\ [5pt]
Let $\epsilon > 0$ and $U_0 \in \cU$ be given. Set $X_{-1}=0$,
$L=\rho^2\beta^2$, $\sigma=1$, and $k=0$.
\begin{itemize}
\item[1)]
Compute $\nabla f(U_k)$ and $X(U_k)$. Set $X_{k} = \frac{k}{k+2} X_{k-1} + 
\frac{2}{k+2} X(U_k)$.
\item[2)]
Find $U^{sd}_{k} = \argmin \left \{ \langle \nabla f(U_k), U-U_k \rangle +
\frac{L}2 \, \|U-U_k\|^2_F: \ U \in \cU \right \}$.
\item[3)]
Find $U^{ag}_{k} = \argmin \left \{ \frac{L}{2\sigma}\|U-U_0\|^2_F+\sum\limits_{i=0}^k
\frac{i+1}2 [f(U_i) + \langle \nabla f(U_i), U-U_i \rangle]: \ U \in \cU \right \}$. 
\item[4)]
Set $U_{k+1} = \frac{2}{k+3} U^{ag}_{k} + \frac{k+1}{k+3} U^{sd}_{k}$.
\item[5)]
Set $k \leftarrow k+1$. Go to step 1) until $f(U^{sd}_k) - g(X_k) \le \epsilon$.
\end{itemize}
\noindent
{\bf end}
\end{minipage}
\vgap

The iteration complexity of the above algorithm for solving problem \eqnok{smooth-saddle} 
is established in the following theorem. 

\begin{theorem} \label{mtm-sparcov}
The iteration complexity performed by the algorithm SMACS for finding an $\epsilon$-optimal 
solution to problem \eqnok{smooth-saddle} and its dual does not exceed 
$\sqrt{2}\rho\beta\max\limits_{U\in\cU}\|U-U_0\|_F/\sqrt{\epsilon}$, and moreover, 
if $U_0=0$, it does not exceed $\sqrt{2}\rho\beta n/\sqrt{\epsilon}$. 
\end{theorem}    

\begin{proof}
From the above discussion, we know that $L = \rho^2\beta^2$, $D=\max\limits_{U\in\cU}
\|U-U_0\|^2_F/2$ and $\sigma=1$, which together with Theorem \ref{mtm-concave} 
immediately implies that the first part of the statement holds. Further, if $U_0=0$, 
we easily obtain from \eqnok{cU} that $D=\max\limits_{U\in\cU} \|U\|^2_F/2=n^2/2$. 
The second part of the statement directly follows from this result and Theorem 
\ref{mtm-concave}.  
\end{proof}

\vgap

\begin{remark}
By the definition of $\cU$ (see \eqnok{cU}), we can easily show that 
$\min\limits_{U_0\in\cU}\max\limits_{U\in\cU} \|U-U_0\|_F$ has a unique 
minimizer $U_0=0$. This result together with Theorem \ref{mtm-sparcov} 
implies that the initial point $U_0=0$ gives the optimal worst-case 
iteration complexity for the algorithm SMACS.  
\end{remark}

\vgap

Alternatively, d'Aspremont et al.\ \cite{DaOnEl06} applied Nesterov's smooth approximation 
scheme \cite{Nest05-1} to solve problem \eqnok{smooth-saddle}. More specifically, let 
$\epsilon > 0$ be the desired accuracy, and let
\[
\hd(U) = \|U\|^2_F/2,  \ \ \ \hD = \max\limits_{U\in \cU} \hd(U) = n^2/2.
\]
As shown in \cite{Nest05-1}, the non-smooth function $g(X)$ defined in \eqnok{gx} is
uniformly approximated by the smooth function 
\[
g_{\epsilon}(X) = \min\limits_{U \in \cU} \ 
\log\det X - \langle \Sigma + \rho U, X \rangle - \frac{\epsilon}{2\hD} \hd(U)
\]
on $\cX$ with the error at most by $\epsilon/2$, and moreover, the function $g_{\epsilon}(X)$ 
has a Lipschitz continuous gradient on $\cX$ with some constant $L(\epsilon)>0$. 
Nesterov's smooth optimization technique \cite{Nest83-1,Nest05-1} is then 
applied to solve the perturbed problem $\max\limits_{X \in \cX} g_{\epsilon}(X)$, and problem 
\eqnok{smooth-saddle} is accordingly solved. 
%which is $\epsilon/2$-perturbation of problem \eqnok{smooth-saddle}.   
It was shown in \cite{DaOnEl06} that the iteration complexity of this approach for finding an 
$\epsilon$-optimal solution to problem \eqnok{smooth-saddle} does not exceed  
\beq \label{sa-comp}
\frac{2\sqrt{2}\rho\beta n^{1.5}\log\kappa}{\epsilon} 
+ \kappa \sqrt{\frac{n\log\kappa}{\epsilon}}
\eeq
where $\kappa := \beta/\alpha$. 

In view of \eqnok{sa-comp} and Theorem \ref{mtm-sparcov}, we conclude that the 
smooth optimization approach improves upon Nesterov's smooth approximation scheme 
at least by a factor $\cO(\sqrt{n}\log \kappa/\sqrt{\epsilon})$ in terms of the 
iteration complexity for solving problem \eqnok{smooth-saddle}. Moreover, the 
computational cost per iteration of the former approach is at least as cheap 
as that of the latter one.
  
d'Aspremont et al.\ \cite{DaOnEl06} also studied a block-coordinate descent method for 
solving problem \eqnok{relax} with $\talpha=0$ and $\tbeta=\infty$. Each iterate 
of this method requires computing the inverse of an $(n-1)\times (n-1)$ matrix, 
and solving a box constrained quadratic programming with $n-1$ variables. 
As mentioned in Section 3 of \cite{DaOnEl06}, this method has a local linear 
convergence rate. However, its global iteration complexity for finding an 
$\epsilon$-optimal solution is theoretically unknown. Moreover, this method is 
not suitable for solving problem \eqnok{relax} with $\talpha>0$ or $\tbeta < \infty$.   

In addition, we observe that problem \eqnok{new-relax} (also \eqnok{relax}) can be reformulated 
as a constrained smooth convex problem that has an explicit $\cO(n^2)$-logarithmically homogeneous 
self-concordant barrier function. Thus, it can be suitably solved by interior point (IP) 
methods (see Nesterov and Nemirovski \cite{NeNe94} and Vandenberghe et al.\ \cite{VaBoWu98}). 
The worst-case iteration complexity of IP methods for finding an $\epsilon$-optimal solution 
to \eqnok{new-relax} is $\cO(n\log(\epsilon_0/\epsilon))$, where $\epsilon_0$ is an initial 
gap. Each iterate of IP methods requires $\cO(n^6)$ arithmetic cost for assembling and 
solving a typically dense Newton system with $\cO(n^2)$ variables.  Thus, the total worst-case 
arithmetic cost of IP methods for finding an $\epsilon$-optimal solution to \eqnok{new-relax} is 
$\cO(n^7\log(\epsilon_0/\epsilon))$. In contrast with IP methods, the algorithm SMACS 
requires $\cO(n^3)$ arithmetic cost per iteration dominated by eigenvalue decomposition 
and matrix multiplication of $n \times n$ matrices. Based on this observation and Theorem 
\ref{mtm-sparcov}, we conclude that the overall worst-case arithmetic cost of the algorithm 
SMACS for finding an $\epsilon$-optimal solution to \eqnok{new-relax} is 
$\cO(\rho\beta n^4/{\sqrt{\epsilon}})$, which is substantially superior to that of 
IP methods, provided that $\rho\beta$ is not too large and $\epsilon$ is not too 
small.       
 
\subsection{Variant of Smooth Minimization Algorithm}
\label{variant}

As discussed in Subsection \ref{appl}, the algorithm SMACS has a nice theoretical 
complexity in contrast with IP methods, Nesterov's smooth approximation scheme, and 
block-coordinate descent method. However, its practical performance is still not much 
attractive (see Section \ref{comp}). To enhance the computational performance, we 
propose a variant of the algorithm SMACS for solving problem \eqnok{smooth-saddle} 
in this subsection.  
 
Our first concern of the algorithm SMACS is that the eigenvalue decomposition of two 
$n \times n$ matrices is required per iteration. Indeed, the eigenvalue 
decomposition of $\Sigma+\rho U_k$ and $\Sigma+\rho U^{sd}_k$ is needed 
at steps 1) and 5) to compute $\nabla f(U_k)$ and $f(U^{sd}_k)$, respectively. 
We also know that the eigenvalue decomposition is one of major computations 
for the algorithm SMACS. To reduce the computational cost, we now propose
a new termination criterion other than $f(U^{sd}_k) - g(X_k) \le \epsilon$ 
that is used in the algorithm SMACS. In view of Theorem \ref{prim-conv}, we 
know that 
\[
f(U_k) - g(X(U_k)) \to 0, \ \mbox{as} \ k \to \infty. 
\] 
Thus, $f(U_k) - g(X(U_k)) \le \epsilon$ can be used as an alternative termination 
criterion. Moreover, it follows from \eqnok{xfu} that the quantity $f(U_k) - g(X(U_k))$ 
is readily available in step 1) of the algorithm SMACS with almost no additional cost. 
We easily see that the algorithm SMACS with this new termination criterion would 
require only one eigenvalue decomposition per iteration. Despite this clear advantage, 
we shall mention that the iteration complexity of the resulting algorithm is unfortunately 
unknown. Nevertheless, in practice we have found that the number of iterations performed by 
the algorithm SMACS with the above two different termination criteria are almost 
same. Thus, $f(u_k) - g(x(u_k)) \le \epsilon$ is a useful practical termination 
criterion. 
 
For sparse covariance selection, the penalty parameter $\rho$ is usually small, 
but the parameter $\beta$ can be fairly large. In view of Theorem \ref{mtm-sparcov}, 
we know that the iteration complexity of the algorithm SMACS for solving problem 
\eqnok{smooth-saddle} is proportional to $\beta$. Therefore, when $\beta$ is too 
large, the complexity and practical performance of this algorithm become unattractive. 
To overcome this drawback, we will propose one strategy to dynamically update $\beta$.   

Let $X^*$ be the unique optimal solution of problem \eqnok{smooth-saddle}. For any 
$\hbeta \in [\lambda_{\max}(X^*), \beta]$, we easily observe that $X^*$ is also the unique 
optimal solution to the following problem:      
\beq \label{smooth-hbeta}
(P_{\hbeta}) \hspace{1.in} \max\limits_{X \in \cXb} \min\limits_{U \in \cU} \ 
\log\det X - \langle \Sigma + \rho U, X \rangle, 
\eeq 
where $\cU$ is defined in \eqnok{cU}, and $\cXb$ is given by
\[
\cXb := \{X: \ \alpha I \preceq X \preceq \hbeta I\}.
\]
In view of Theorem \ref{mtm-sparcov}, the iteration complexity of the algorithm SMACS 
for problem \eqnok{smooth-hbeta} is lower than that for problem \eqnok{smooth-saddle} 
provided $\hbeta \in [\lambda_{\max}(X^*), \beta)$. Hence ideally, we set 
$\hbeta=\lambda_{\max}(X^*)$, which would give the lowest iteration complexity, 
but unfortunately, $\lambda_{\max}(X^*)$ is unknown. However, we can generate a 
sequence $\{\hbeta_k\}^{\infty}_{k=0}$ that asymptotically approaches $\lambda_{\max}(X^*)$ 
as the algorithm progresses. Indeed, in view of Theorem \ref{prim-conv}, we know that 
$X(U_k) \to X^*$ as $k \to \infty$, and we obtain that 
\[
\lambda_{\max}(X(U_k)) \to \lambda_{\max}(X^*), \ \mbox{as} \ k \to \infty.
\]     
%Thus, $\lambda_{\max}(X^*)$ can be asymptotically estimated by the sequence 
%$\{\lambda_{\max}(X(U_k))\}^{\infty}_{k=0}$. Based on this observation, 
Therefore, we see that $\{\lambda_{\max}(X(U_k))\}^{\infty}_{k=0}$ can be used to 
generate a sequence $\{\hbeta_k\}^{\infty}_{k=0}$ that asymptotically approaches 
$\lambda_{\max}(X^*)$. We next propose a strategy to generate such a sequence 
$\{\hbeta_k\}^{\infty}_{k=0}$.

For convenience of presentation, we introduce some new notations. Given any 
$U\in \cU$ and $\hbeta \in [\alpha, \beta]$, we define  
\beqa
X_{\hbeta}(U) &:=& \arg\max\limits_{X \in \cXb} \ \log\det X - \langle \Sigma + \rho U, 
 X \rangle, \label{xbeta} \\
f_{\hbeta}(U) &:=& \max\limits_{X \in \cXb} \ \log\det X - \langle \Sigma + \rho U, 
 X \rangle. \label{fbeta}  
\eeqa 
 
\begin{definition} 
Given any $U\in \cU$ and $\hbeta \in [\alpha,\beta]$, $X_{\hbeta}(U)$ is called ``active'' 
if $\lambda_{\max}(X_{\hbeta}(U))=\hbeta$ and $\hbeta < \beta$; otherwise it is 
called ``inactive''. 
\end{definition}

Let $\varsigma_1$, $\varsigma_2 > 1$, and $\varsigma_3 \in (0,1)$ be 
given and fixed. Assume that $U_k \in\cU$ and $\hbeta_k \in [\alpha, \beta]$ are given
at the beginning of the $k$th iteration for some $k \geq 0$. We now describe the strategy 
for generating the next iterate $U_{k+1}$ and $\hbeta_{k+1}$ by considering the following 
three different cases: 
\bi
\item[1)] 
If $\xbuk$ is active, find the smallest $s\in \cZ_{+}$ such that $X_{\bbeta}(U_k)$ is 
inactive, where $\bbeta=\min\{\varsigma_1^s \hbeta_k, \beta\}$. Set 
$\hbeta_{k+1} = \bbeta$, and apply the algorithm SMACS for problem $(P_{\hbeta_{k+1}})$ 
starting with the point $U_k$ and set its next iterate to be $U_{k+1}$.
\item[2)]
If $\xbuk$ is inactive and $\lmax(\xbuk) \le \varsigma_3 \hbeta_k$, set $\hbeta_{k+1} = 
\max\{\min\{\varsigma_2\lmax(\xbuk),\\ \beta\},\alpha\}$. Apply the algorithm SMACS for problem 
$(P_{\hbeta_{k+1}})$ starting with the point $U_k$, and set its next iterate to be $U_{k+1}$.
\item[3)]
If $\xbuk$ is inactive and $\lmax(\xbuk) > \varsigma_3 \hbeta_k$, set $\hbeta_{k+1} = \hbeta_k$.
Continue the algorithm SMACS for problem $(P_{\hbeta_k})$, and set its next iterate to be 
$U_{k+1}$.
\ei
  
For the sequences $\{U_k\}^{\infty}_{k=0}$ and $\{\hbeta_k\}^{\infty}_{k=0}$ recursively 
generated above, we observe that the sequence $\{X_{\hbeta_{k+1}}(U_k)\}^{\infty}_{k=0}$ 
is always inactive. This together with \eqnok{xbeta}, \eqnok{fbeta}, \eqnok{fU} and 
the fact that $\hbeta_k \le \beta$ for $k \ge 0$, implies that
\beq \label{fgbeta}
f(U_k) = f_{\hbeta_{k+1}}(U_k), \ \ \ \nabla f(U_k) = \nabla f_{\hbeta_{k+1}}(U_k), \ \
\forall k \ge 0.
\eeq
Therefore, the new termination criterion $f(U_k) - g(X(U_k)) \le \epsilon$ can be 
replaced by 
\beq \label{newterm}
f_{\hbeta_{k+1}}(U_k) - g(X_{\hbeta_{k+1}}(U_k)) \le \epsilon
\eeq 
accordingly.    

We now incorporate into the algorithm SMACS the new termination criterion \eqnok{newterm} 
and the aforementioned strategy for generating a sequence $\{\hbeta_k\}^{\infty}_{k=0}$ 
that asymptotically approaches $\lmax(X^*)$, and obtain a variant of the algorithm SMACS 
for solving problem \eqnok{smooth-saddle}. For convenience of presentation, we omit the 
subscript $k$ from $\hbeta_k$.

\gap

\noindent
\begin{minipage}[h]{6.6 in}
{\bf Variant of Smooth Minimization Algorithm for Covariance Selection (VSMACS):} \\ [5pt]
Let $\epsilon > 0$, $\varsigma_1$, $\varsigma_2 > 1$, and $\varsigma_3 \in (0,1)$ be given. 
Choose a $U_0 \in \cU$. Set $\hbeta = \beta$, $L=\rho^2\beta^2$, $\sigma=1$, and $k=0$.
\begin{itemize}
\item[1)] Compute $\xhbuk$ according to \eqnok{xfu}.
\bi
\item[1a)] 
If $\xhbuk$ is active, find the smallest $s\in \cZ_{+}$ such that $X_{\bbeta}(U_k)$ is 
inactive, where $\bbeta=\min\{\varsigma_1^s \hbeta, \beta\}$. Set $k=0$, $U_0 = U_k$, 
$\hbeta = \bbeta$, $L=\rho^2\hbeta^2$, and go to step 2). 
\item[1b)]
If $\xhbuk$ is inactive and $\lmax(\xhbuk) \le \varsigma_3 \hbeta$, set $k=0$, $U_0 = U_k$, 
\\ $\hbeta=\max\{\min\{\varsigma_2 \lmax(\xhbuk),\beta\},\alpha\}$, and $L=\rho^2\hbeta^2$. 
\ei
\item[2)] If $f_{\hbeta}(U_k) - g(\xhbuk) \le \epsilon$, terminate.
Otherwise, compute $\nabla f_{\hbeta}(U_k)$ according to \eqnok{xfu}.
%and set $\nabla f(U_k) = \nabla f_{\hbeta}(U_k)$.
\item[3)]
Find $U^{sd}_{k} = \argmin \left \{ \langle \nabla f_{\hbeta}(U_k), U-U_k \rangle +
\frac{L}2 \, \|U-U_k\|^2_F: \ U \in \cU \right \}$.
\item[4)]
Find $U^{ag}_{k} = \argmin \left \{ \frac{L}{2\sigma}\|U-U_0\|^2_F+\sum\limits_{i=0}^k
\frac{i+1}2 [f_{\hbeta}(U_i)+ \langle \nabla f_{\hbeta}(U_i), U-U_i \rangle]: \ U \in \cU \right \}$. 
\item[5)]
Set $U_{k+1} = \frac{2}{k+3} U^{ag}_{k} + \frac{k+1}{k+3} U^{sd}_{k}$.
\item[6)]
Set $k \leftarrow k+1$, and go to step 1). 
\end{itemize}
\noindent
{\bf end}
\end{minipage}
\vgap
    
We next establish some preliminary convergence properties of the above algorithm.

\begin{proposition} \label{prop-conv}
For the algorithm VSMACS, the following properties hold:
\bi
\item[1)] 
Suppose that the algorithm VSMACS terminates at some iterate $(\xhbuk, U_k)$. Then 
$(\xhbuk$, $U_k)$ is an $\epsilon$-optimal solution to problem \eqnok{smooth-saddle} 
and its dual.
\item[2)]
Suppose that $\hbeta$ is updated only for a finite number of times. Then the algorithm 
VSMACS terminates in a finite number of iterations, and produces an $\epsilon$-optimal 
solution to problem \eqnok{smooth-saddle} and its dual. 
\ei 
\end{proposition}

\begin{proof}
For the final iterate $(\xhbuk, U_k)$, we clearly know that $f_{\hbeta}(U_k) - g(\xhbuk) 
\le \epsilon$, and $\xhbuk$ is inactive. As shown in \eqnok{fgbeta}, $f(U_k) = 
f_{\hbeta}(U_k)$. Hence, we have $f(U_k)-g(\xhbuk) \le \epsilon$. We also know that 
$U_k \in \cU$, and $\xhbuk \in \cX$ due to $\hbeta \in [\alpha, \beta]$. Thus, 
statement 1) immediately follows. After the last update of $\hbeta$, the algorithm 
VSMACS behaves exactly like the algorithm SMACS as applied to solve problem 
$(P_{\hbeta})$ except with the termination criterion $f(U_k)-g(\xhbuk) \le \epsilon$. 
Thus, it follows from statement 1) and Theorem \ref{prim-conv} that statement 2) holds.   
\end{proof}

\section{Computational results}
\label{comp}

In this section, we compare the performance of the smooth minimization approach and 
its variant proposed in this paper with other first-order methods studied in 
\cite{DaOnEl06,FriHasTib07}, that is, Nesterov's smooth approximation 
scheme and block coordinate descent method for solving problem \eqnok{relax} 
(or equivalently, \eqnok{smooth-saddle}) on a set of randomly generated instances. 

All instances used in this section were randomly generated in the same manner 
as described in d'Aspremont et al.\ \cite{DaOnEl06}. First, we generate a sparse 
invertible matrix $A\in \cS^n$ with positive diagonal entries and a density 
prescribed by $\varrho$. 
%and a few randomly chosen nonzero off-diagonal entries. 
We then generate the matrix $B\in\cS^n$ by
\[
B = A^{-1} + \tau V,
\]      
where $V\in\cS^n$ is an independent and identically distributed uniform random 
matrix, and $\tau$ is a small positive number. Finally, we obtain the following 
randomly generated sample covariance matrix:
\[
\Sigma = B - \min\{\lambda_{\min}(B)- \vartheta, 0\} I,   
\]
where $\vartheta$ is a small positive number. In particular, we set $\varrho=0.01$, 
$\tau=0.15$ and $\vartheta=1.0e-4$ for generating all instances. 

As discussed in Section \ref{appl}, our smooth minimization approach has much 
better worst-case iteration complexity than Nesterov's smooth approximation 
scheme studied in d'Aspremont et al.\ \cite{DaOnEl06} for problem \eqnok{smooth-saddle}. 
However, it is unknown how their practical performance differs from each other.  
In the first experiment, we compare the practical performance of 
%the smooth minimization approach and its variant proposed in this paper 
our smooth minimization approach and its variant with Nesterov's smooth approximation 
scheme studied in d'Aspremont et al.\ \cite{DaOnEl06} for problem \eqnok{smooth-saddle} 
with $\alpha=0.1$, $\beta=10$ and $\rho=0.5$. For convenience of presentation, we label 
these three first-order methods as SM, VSM, and NSA, respectively. The codes for them 
are written in Matlab. More specifically, the code for NSA follows the algorithm 
presented in d'Aspremont et al.\ \cite{DaOnEl06}, and the codes for SM and VSM are 
written in accordance with the algorithms SMACS and VSMACS, respectively. Moreover, 
we set $\varsigma_1= \varsigma_2 = 1.05$ and $\varsigma_3=0.95$ for the algorithm 
VSMACS. These three methods terminate once the duality gap is less than $\eps=0.1$. 
All computations are performed on an Intel Xeon 2.66 GHz machine with Red Hat Linux 
version 8. 

The performance of the methods NSA, SM and VSM for the randomly generated instances 
are presented in Table \ref{result-1}. The row size $n$ of each sample covariance matrix  
$\Sigma$ is given in column one. The numbers of iterations of NSA, SM and VSM 
are given in columns two to four, and the objective function values are given in columns 
five to seven, and the CPU times (in seconds) are given in the last three columns, 
respectively. From Table \ref{result-1}, we conclude that: i) the method SM, namely, 
the smooth minimization approach, outperforms substantially  the method NSA, that is, 
Nesterov's smooth approximation scheme; and ii) the method VSM, namely, the variant of 
the smooth minimization approach, substantially outperforms the other two methods. In 
addition, we see from this experiment that Nesterov's smooth minimization approach 
\cite{Nest83-1} is generally more appealing than his smooth approximation scheme 
\cite{Nest05-1} whenever the problem can be solved as an equivalent smooth problem. 
Nevertheless, we shall mention that the latter approach has much wider field of 
application (e.g., see \cite{Nest05-1}), where the former approach cannot be 
applied.    

\begin{table}[t]
\caption{Comparison of NSA, SM and VSM}
\centering
\label{result-1}
%\begin{center}
\begin{small}
\begin{tabular}{|c||rrr||rrr||rrr|}
\hline 
\multicolumn{1}{|c||}{Problem} & \multicolumn{3}{c||}{Iter} &  
\multicolumn{3}{c||}{Obj} & \multicolumn{3}{c|}{Time} \\
\multicolumn{1}{|c||}{n} & \multicolumn{1}{c}{\sc nsa} 
& \multicolumn{1}{c}{\sc sm} & \multicolumn{1}{c||}{\sc vsm}  
& \multicolumn{1}{c}{\sc nsa} & \multicolumn{1}{c}{\sc sm} 
& \multicolumn{1}{c||}{\sc vsm} & \multicolumn{1}{c}{\sc nsa} 
& \multicolumn{1}{c}{\sc sm} & \multicolumn{1}{c|}{\sc vsm} \\
\hline
50 & 3657 & 457 & 20  & -76.399 & -76.399  & -76.393  & 49.0 & 2.7 & 0.1 \\
100 & 7629 & 920 & 27 & -186.717 & -186.720 & -186.714 & 900.4 & 38.4 & 0.4 \\
150 & 20358 & 1455 & 49 & -318.195 & -318.194 & -318.184 & 8165.7 & 188.8 & 2.0 \\
200 & 27499 & 2294 & 102 & -511.246 & -511.245 & -511.242 & 26172.5 & 698.8 & 9.2 \\
250 & 45122 & 3060 & 128 & -3793.255 & -3793.256 & -3793.257 & 87298.9 & 1767.9 & 19.8 \\
300 & 54734 & 3881 & 161 & -3187.163 &  -3187.171 & -3187.172 & 184798.1 & 3994.0 & 45.5 \\
350 & 64641 & 4634 & 182 & -2756.717 & -2756.734  & -2756.734 & 351460.7 & 7613.9 & 83.6 \\
400 & 74839 & 5308 & 176 & -3490.640 & -3490.667  & -3490.667 & 614237.1 & 13536.7 & 116.9 \\
%450 & 85948 & 5937 & 177 & -3631.003  & -3631.044 & -3631.043 & %1013189.4 & 26006.2 & 251.7 \\
\hline
\end{tabular}
\end{small}
%\end{center}
\end{table}

 From the above experiment, we have already seen that the method VSM outperforms 
substantially two other first-order methods, namely, SM and NSA for solving 
problem \eqnok{smooth-saddle}. In the second experiment, we compare the 
performance of the method VSM with the block coordinate descent (BCD) methods 
studied in d'Aspremont et al.\ \cite{DaOnEl06} and Friedman et al.\ 
\cite{FriHasTib07} on relatively large-scale instances. For convenience of 
presentation, we label these two methods as BCD1 and BCD2, respectively. The 
method BCD2 was developed very recently and it is a slight variant of the method 
BCD1. In particular, each iterate of BCD1 solves a box constrained quadratic 
programming by means of interior point methods, but each iterate of BCD2 applies 
a coordinate descent approach to solving a lasso ($l_1$-regularized) least-squares 
problem, which is the dual of the box constrained quadratic programming appearing 
in BCD1. It is worth mentioning the methods BCD1 and BCD2 are only applicable 
for solving problem \eqnok{relax} with $\talpha=0$ and $\tbeta=\infty$. Thus, we 
only compare their performance with our method VSM for problem \eqnok{relax} with such 
$\talpha$ and $\tbeta$. As shown in Subsection \ref{reform}, problem \eqnok{relax} 
with $\talpha=0$ and $\tbeta=\infty$ is equivalent to problem \eqnok{smooth-saddle} 
with $\alpha$ and $\beta$ given in \eqnok{alpha}, and hence it can be solved by 
applying the method VSM to the latter problem instead. 

The code for the method BCD1 was written in Matlab by d'Aspremont and El Ghaoui 
\cite{COVSEL06} while the code for BCD2 was written in Fortran 90 by Friedman 
et al.\ \cite{FriHasTib07-1}. The methods BCD1 and VSM terminate once the duality 
gap is less than $\eps=0.1$. The original code \cite{FriHasTib07-1} for BCD2 uses 
the average absolute change in the approximate solution as the termination criterion. 
In particular, the average absolute change in the approximate solution is evaluated 
at the end of each cycle consisting of $n$ block coordinate descent iterations, 
and their code terminates once it is below a given accuracy (see pp.\ 6 of \cite{FriHasTib07} 
for details). According to our computational experience, we found with such a  
criterion, BCD2 is extremely hard to terminate for relatively large-scale 
instances (say $n=300$) unless a maximum number of iterations is set. Obviously, 
it is not easy to choose a suitable maximum number of iterations for BCD2. Thus, 
to be as fair as possible to BCD1 and VSM, we simply replace their termination 
criterion detailed in \cite{FriHasTib07-1} for BCD2 by the one with the duality 
gap less than $\eps=0.1$. In other words, the duality gap is computed at the 
end of each cycle consisting of $n$ block coordinate descent iterations, and 
BCD2 terminates once it is below $\eps=0.1$. It is worth remarking that the 
cost for computing a duality gap is ${\cal O}(n^3)$ since the inverse of an 
$n \times n$ symmetric matrix is needed. Thus, it is reasonable to compute 
duality gap once every $n$ iterations rather than each iteration. 

All computations are performed on an Intel Xeon 2.66 GHz machine with Red Hat 
Linux version 8. The performance of the methods BCD1, BCD2 and VSM for the 
randomly generated instances are presented in Table \ref{result-2}. The row 
size $n$ of each sample covariance matrix  $\Sigma$ is given in column one. 
The numbers of iterations of BCD1, BCD2 and VSM are given in columns two to 
four, and the objective function values are given in columns five to eight, 
and the CPU times (in seconds) are given in the last three columns, 
respectively. From Table \ref{result-2}, we conclude both BCD2 and VSM 
substantially outperform BCD1. We also observe that our method VSM 
outperforms BCD2 for almost all instances except two relatively small-scale 
instances.    
%It shall be mentioned that BCD and VSM are both feasible methods, and moreover, 
%\eqnok{relax} and \eqnok{smooth-saddle} are maximization problems. Therefore 
%for these two methods, the larger objective function value, the better. From Table 
%\ref{result-2}, we conclude that the method VSM, namely, the variant of the smooth 
%minimization approach, substantially outperforms the BCD, that is, block coordinate 
%descent method, and it also produces better objective function values for most of 
%the instances.

\begin{table}[t]
\caption{Comparison of BCD1, BCD2 and VSM}
\centering
\label{result-2}
%\begin{center}
\begin{small}
\begin{tabular}{|c||rrr||rrr||rrr|}
\hline 
\multicolumn{1}{|c||}{Problem} & \multicolumn{3}{c||}{Iter} &  
\multicolumn{3}{c||}{Obj} & \multicolumn{3}{c|}{Time} \\
\multicolumn{1}{|c||}{n} & \multicolumn{1}{c}{\sc bcd1} 
& \multicolumn{1}{c}{\sc bcd2} & \multicolumn{1}{c||}{\sc vsm} 
& \multicolumn{1}{c}{\sc bcd1} & \multicolumn{1}{c}{\sc bcd2} 
& \multicolumn{1}{c||}{\sc vsm} & \multicolumn{1}{c}{\sc bcd1} 
& \multicolumn{1}{c}{\sc bcd2} & \multicolumn{1}{c|}{\sc vsm} \\
\hline
100 & 124 & 200 & 33 & -186.522 & -186.433 & -186.522  & 22.3 & 0.1 & 0.5 \\
200 & 531 & 600 & 109 & -449.210 & -449.179 & -449.209 & 300.0 & 1.3 & 9.5 \\
300 & 1530 & 1500 & 146 & -767.615 & -767.608 & -767.614 & 2428.2 & 80.9 & 48.5 \\
400 & 2259 & 2400 & 154 & -1082.679 & -1082.651 & -1082.677 & 8402.4 & 298.7 & 112.3 \\
500 & 3050 & 3500 & 154 & -1402.503 & -1402.457 & -1402.502 & 22537.1 & 640.2 & 211.5 \\
600 & 3705 & 4200 & 165 & -1728.628 & -1728.587 & -1728.627 & 48950.4 & 1215.0 & 397.6 \\
700 & 4492 & 4900 & 163 & -2057.894 & -2057.862 & -2057.892 & 92052.7 & 1972.5 & 611.1 \\
800 & 4958 & 5600 & 169 & -2392.713 & -2392.671 & -2392.712 & 147778.9 & 2872.3 & 943.2 \\
900 & 5697 & 6300 & 161 & -2711.874 & -2711.827  & -2711.874 & 219644.3 & 3593.7 & 1268.5 \\
1000 & 6536 & 7000 & 161 & -3045.808 & -3045.768 & -3045.808 & 344687.8 & 6098.7 & 1710.0 \\
\hline
\end{tabular}
\end{small}
%\end{center}
\end{table}

In the above experimentation, we compared the performance of BCD2 and VSM for 
$\eps=0.1$. We next compare their performance on the same instances as above and 
apply the same termination criterion as above except that we set up $\eps=0.01$ and 
an upper bound of $20$ hours computation time (or $72,000$ seconds) per instance 
for both codes. The performance of the methods BCD2 and VSM are presented in Table 
\ref{result-3}. The row size $n$ of each sample covariance matrix  $\Sigma$ is 
given in column one. The numbers of iterations of BCD2 and VSM are given in 
columns two to three, and the objective function values are given in columns 
four to five, and the CPU times (in seconds) are given in the last two columns, 
respectively. It shall be mentioned that BCD2 and VSM are both feasible methods, 
and moreover, \eqnok{relax} and \eqnok{smooth-saddle} are maximization problems. 
Therefore for these two methods, the larger objective function value, the better. 
From Table \ref{result-3}, we observe that up to accuracy $\eps=0.01$, the method 
BCD2 cannot solve almost all instances within $20$ hours except the first three 
relatively small-scale ones, but our method VSM does solve each of these instances 
in less than one hour and produces a better objective function values for almost 
all instances except the first three relatively small-scale ones. Also, it is 
interesting to observe that the number of iterations for VSM nearly doubles as 
the accuracy parameter $\eps$ increases by one digit, which is even better than 
the theoretical estimate that is $\sqrt{10}$ according to Theorem \ref{mtm-sparcov}.  

\begin{table}[t]
\caption{Comparison of BCD2 and VSM}
\centering
\label{result-3}
%\begin{center}
\begin{small}
\begin{tabular}{|c||rr||rr||rr|}
\hline 
\multicolumn{1}{|c||}{Problem} & \multicolumn{2}{c||}{Iter} &  
\multicolumn{2}{c||}{Obj} & \multicolumn{2}{c|}{Time} \\
\multicolumn{1}{|c||}{n} & \multicolumn{1}{c}{\sc bcd2} 
& \multicolumn{1}{c||}{\sc vsm} & \multicolumn{1}{c}{\sc bcd2} 
& \multicolumn{1}{c||}{\sc vsm} & \multicolumn{1}{c}{\sc bcd2} 
& \multicolumn{1}{c|}{\sc vsm} \\
\hline
100 & 200 & 54 & -186.433 & -186.435  & 0.1 & 0.77 \\
200 & 1200 & 239 & -449.119 & -449.122 & 2.1 & 21.6 \\
300 & 3000 & 310 & -767.525 & -767.525 & 32.1 & 104.2 \\
400 & 11778400 & 321 & -1082.592 & -1082.589 & 72000.0 & 223.3 \\
500 & 6997000 & 309 & -1402.420& -1402.413 &  72001.0 & 395.5 \\
600 & 4637400 & 318 & -1728.553 & -1728.538 & 72004.0 & 765.2 \\
700 & 3215100 & 310 & -2057.823 & -2057.804 & 72005.0 & 1330.0 \\
800 & 2307200 & 309 & -2392.644 & -2392.623 & 72003.0 & 1789.2 \\
900 & 1846800 & 289 & -2711.806 & -2711.784 & 72024.0 & 2394.0 \\
1000 & 1257000 & 283 & -3045.749 & -3045.718 & 72051.0 & 3115.8 \\
\hline
\end{tabular}
\end{small}
%\end{center}
\end{table}

\section{Concluding remarks}
\label{concl-remark}

In this paper, we proposed a smooth optimization approach for solving a 
class of non-smooth strictly concave maximization problems. We also discussed 
the application of this approach to sparse covariance selection, and proposed a variant 
of this approach. The computational results showed that the variant of the smooth 
optimization approach outperforms substantially the latter one, and two other 
first-order methods studied in d'Aspremont et al.\ \cite{DaOnEl06} and Friedman 
et al.\ \cite{FriHasTib07}.   

As discussed in Subsection \ref{appl}, problem \eqnok{smooth-saddle} has the same 
form as \eqnok{concave-opt}, and satisfies all assumptions imposed on problem 
\eqnok{concave-opt}. Moreover, its associated objective function 
$\phi(X,U)=\log\det X - \langle \Sigma + \rho U, X \rangle$ is affine with respect 
to $U$ for every fixed $X \in \cS^n_{++}$. In view of these facts along with the 
remarks made in Section \ref{smooth-appr}, one can observe that problem 
\eqnok{smooth-saddle} can be suitably solved by Nesterov's excessive gap technique 
\cite{Nest05-2}. Since the iterate complexity and the computational cost 
per iterate of this technique is same as those of the algorithm SMACS, we expect that the 
computational performance of these two methods for solving \eqnok{smooth-saddle} are 
similar. It would be interesting to implement Nesterov's excessive gap technique 
\cite{Nest05-2} and its variant (that is, the one in a similar fashion to 
the algorithm VSMACS), and compare their computational performance with 
SMACS and VSMACS, respectively. 

%After the first release of our paper, Friedman et al.\ \cite{FriHasTib07} 
%studied a slight variant of the block coordinate descent (BCD) method proposed in 
%\cite{DaOnEl06} for solving problem \eqnok{smooth-saddle} with $\talpha=0$ and 
%$\tbeta=\infty$. At each iteration of their method, a coordinate descent approach 
%is applied to solve a lasso ($l_1$-regularized) least-squares problem, which 
%is the dual of the box constrained quadratic programming appearing in the BCD 
%method discussed in \cite{DaOnEl06}. It is not hard to observe that when the 
%sample covariance matrix $\Sigma$ is sparse, the computational cost of each 
%iteration of the variant method \cite{FriHasTib07} might be cheaper 
%than that of the BCD method \cite{DaOnEl06} since it makes efficient 
%use of the sparsity. However, for a general dense sample covariance matrix 
%$\Sigma$, the computational cost per iteration of these two methods cannot be 
%compared directly. In addition, though Friedman et al.\ \cite{FriHasTib07} 
%reported that their method substantially outperforms the BCD method of 
%\cite{DaOnEl06} on some special instances, it is worth mentioning that 
%different termination criteria were used for these two methods in their 
%implementations, and thus their comparison of these two methods is not very 
%convincing. Similarly, it is not easy to compare their method \cite{FriHasTib07} 
%with our smooth optimization approach directly. Instead, we would like to compare 
%the performance of these two methods extensively on real data sets in the future.       

Though the variant of the smooth optimization approach outperforms substantially 
the smooth optimization approach, we are currently only able to establish 
some preliminary convergence properties for it. A possible direction leading to 
a thorough proof of its convergence would be to show that the updates on $\hbeta$ 
in the algorithm VSMACS can occur only for a finite number of times. Given that 
VSMACS is a nonmonotone algorithm, it is, however, highly challenging to analyze 
the behavior of the sequences $\{U_k\}$ and $\{X_{\hbeta}(U_k)\}$ and hence the 
total number of updates on $\hbeta$. Interestingly, we observed in our implementation 
that when $\hbeta > \lambda_{\max}(X^*)$, the sequence $\{X_{\hbeta}(U_k)\}$ 
generated by the algorithm VSMACS satisfies $\lambda_{\max}(X_{\hbeta}(U_k)) 
\in [\lambda_{\max}(X^*), \hbeta)$, where $X^*$ is the optimal solution of problem 
\eqnok{smooth-saddle}. Nevertheless, it remains completely open whether this holds in 
general or not.  In addition, the ideas used in the variant of the smooth 
optimization approach are interesting in their own right even when viewed as 
some heuristics. They could also be used to enhance the practical performance 
of Nesterov's first-order methods \cite{Nest83-1,Nest05-1} for solving some 
general min-max problems.      

The codes for the variant of the smooth minimization approach 
%for problem \eqnok{relax} 
are written in Matlab and C, which are available online at www.math.sfu.ca/$\sim$zhaosong. 
The C code for this method can solve large-scale problems more efficiently provided that 
LAPACK package is suitably installed. We will plan to extend these codes for 
solving more general problems of the form 
\[
\begin{array}{ll}
\max\limits_X & \log\det X - \langle \Sigma, X \rangle - \sum\limits_{ij} 
\omega_{ij} |X_{ij}|  \\ 
\mbox{s.t.} & \talpha I \preceq X \preceq \tbeta I, \\ [4pt]
            & X_{ij} = 0, \ \forall (i,j) \in \Omega,  
\end{array}   
\]
for some set $\Omega$, where $\omega_{ij}=\omega_{ji} \ge 0$ for all $i,j=1,\ldots,n$, 
and $0 \le \talpha < \tbeta \le \infty$ are some fixed bounds on the eigenvalues of the 
solution. 
 
\section*{Acknowledgement}
The author would like to thank Professor Alexandre d'Aspremont for a careful 
discussion on the iteration complexity of Nesterov's smooth approximation scheme 
for sparse covariance selection given in \cite{DaOnEl06}. Also, the author is in debt 
to two anonymous referees for numerous insightful comments and suggestions, 
which have greatly improved the paper.

\end{document}